\documentclass[conference,final,onecolumn,10pt,twoside, narroweqnarray]{IEEEtran}
\usepackage{color, enumitem}
\usepackage{graphicx}
\usepackage{amsmath,amssymb,amsthm}
\usepackage{latexsym}
\usepackage{cite}
\newtheorem{prop}{Proposition}

\newtheorem{lemma}{Lemma}
\newtheorem{theorem}{Theorem}

\newcommand{\R}{\mathbf{R}}
\newcommand{\V}{\mathbf{V}}
\newcommand{\T}{T^E}

\allowdisplaybreaks

\ifCLASSINFOpdf
   \graphicspath{{../jpg/}{../pdf/}{../jpeg/}{../eps/}}
   \DeclareGraphicsExtensions{.jpg,.pdf,.jpeg,.png,.eps}
\else
  \DeclareGraphicsExtensions{.eps}
\fi

\ifCLASSOPTIONcompsoc
  \usepackage[caption=false,font=normalsize,labelfont=sf,textfont=sf]{subfig}
\else
  \usepackage[caption=false,font=footnotesize]{subfig}
\fi
\ifCLASSOPTIONcompsoc
\else
\fi

\usepackage{setspace}

\begin{document}
%
\title{Energy Efficiency in Two-Tiered \\ Wireless Sensor Networks\vspace{-10pt}}


\author{Jun Guo$^{\star}$, Erdem Koyuncu$^{\dagger}$, and Hamid Jafarkhani$^{\star}$ \\
 $^{\star}$Center for Pervasive Communications \& Computing,  University of California, Irvine \\  $^{\dagger}$Department of Electrical and Computer Engineering,  University of Illinois at Chicago\vspace{-0pt}}

\maketitle

\begin{abstract}
We study a two-tiered wireless sensor network (WSN) consisting of $N$ access points (APs) and $M$ base stations (BSs). The sensing data, which is distributed on the sensing field according to a density function $f$, is first transmitted to the APs and then forwarded to the BSs. Our goal is to find an optimal deployment of APs and BSs to minimize the average weighted total, or Lagrangian, of sensor and AP powers. For $M=1$, we show that the optimal deployment of APs is simply a linear transformation of the optimal $N$-level quantizer for density $f$, and the sole BS should be located at the geometric centroid of the sensing field. Also, for a one-dimensional network and uniform $f$, we determine the optimal deployment of APs and BSs for any $N$ and $M$. Moreover, to numerically optimize node deployment for general scenarios, we propose one- and two-tiered Lloyd algorithms and analyze their convergence properties. Simulation results show that, when compared to random deployment, our algorithms can save up to 79\% of the power on average.
\end{abstract}


%
\IEEEpeerreviewmaketitle

\section{Introduction}
\label{secIntro}
Energy efficiency is a key issue in wireless sensor networks (WSNs) as most sensors have limited battery life, and it is inconvenient or even infeasible to replenish the batteries of numerous densely deployed sensors. In general, the energy consumption of a device includes communication energy, data processing energy and sensing energy (for sensors). Data processing may include source coding and channel coding and thus consumes a lot of energy \cite{HYHJ}. On the other hand, experimental measurements show that, in many applications, data processing consumes significantly less energy compared to communication \cite{GA, MAR}. Furthermore, for passive sensors, such as light sensors and acceleration sensors, the sensing energy is negligible. Therefore, there are many applications in which communication energy consumption is dominant.

WSNs can be divided into two classes based on their network architectures: non-hierarchical (or flat) WSNs and hierarchical WSNs. In non-hierarchical WSNs, every sensor has the same role and functionality. The connectivity of WSNs is maintained by multi-hop communication among the sensor nodes.
In contrast, hierarchical networks are often divided into clusters. The cluster heads have more powerful functionalities compared to ``regular'' sensor nodes, which connect to their cluster heads by one-hop communication.
In this paper, we focus on the radio communication energy consumption on such hierarchical networks which are relevant in both WSNs \cite{HLJW,WP,JSYKT,MBK} and cellular networks \cite{AHS,HEEH}.

A huge body of literature exists on reducing energy consumed by radio communication. Three kinds of protocols, (i) power control, (ii) routing, and (iii) topology control, are applied to achieve this task. The power control protocols \cite{VK,SN} control the power level at each node while keeping the connectivity. Also, there has been many studies on the connectivity of wireless networks \cite{HYHJ2,EKHJ}.
The routing protocols \cite{KA,YX} attempt to find an optimal path to transfer data. The topology control protocols \cite{XLYM,ZXQW} avoid unnecessary energy consumption by switching node states (awake or asleep). However, to implement the above protocols, nodes require the knowledge of their location via GPS-capable antennas or via message exchange \cite{OYSF}. Clustering \cite{OYSF,VK} is another useful tool to reduce energy consumption and prolong the network lifetime, which requires information exchanges between the neighboring access points (APs). Therefore, most clustering algorithms, such as HEED clustering \cite{OYSF}, require message exchange between the nodes, which adds to energy consumption. Besides, the above energy saving approaches can only be applied after sensor placement. While there has been a lot of work on sensor deployment, including our own work \cite{GJ}, and the AP deployment for energy efficiency in \cite{MG}, to the best of our knowledge,  the optimal energy-efficient base station (BS) and access
point deployment has not been considered in the literature.

In this paper, we study the node deployment problem in two-tiered WSNs consisting of multiple access points and base stations (also called fusion centers). Our goal is to find the optimal deployment of APs and BSs to minimize the total communication energy consumption of the network. We find the optimal deployments and the corresponding minimum powers for certain special cases. We also propose numerical Lloyd-like algorithms to minimize energy consumption in general.

The rest of this paper is organized as follows: In Section \ref{sec:model}, we  introduce the system model. In Section \ref{sec:opt1}, we find the optimal deployment of APs with one BS. In Section \ref{sec:opt2}, we study the optimal deployment problem with multiple APs and BSs. In Section \ref{sec:algorithm}, we propose numerical algorithms to minimize the energy consumption. In Section \ref{sec:simulation}, we present numerical simulations. Finally, in
Section \ref{sec:conclusion}, we draw our main conclusions. {\color{black}Some of the technical proofs are provided in the appendices.}
\section{System Model and Problem Formulation}\label{sec:model}
We consider a two-tiered WSN consisting of three kinds of nodes: sensors, APs, and BSs. As discussed in Section \ref{secIntro}, there are many examples for which the energy consumption of the sensors is dominated by communication energy. In these cases, sensors are also usually small and relatively cheap. Hence, it is plausible that many sensors are deployed in the target region. On one hand, these densely distributed sensor nodes are grouped into clusters and transfer data to the local APs via single-hop communication. On the other hand, APs collect data from clusters and forward the data to the associated BSs.
Under such two-tiered architecture, the complex routing protocol and the corresponding message exchanges can be avoided. A similar network architecture has been studied in \cite{YTH} where the authors ignore the energy consumption of the sensor nodes.
In what follows, we model the energy consumption by radio communication in two-tiered WSNs.


Let $\Omega$ be a convex polygon in $\Re^{2}$ including its interior. Given the target area $\Omega$, $N$ APs and $M$ BSs are deployed to collect data. $\mathcal{I}_{A}=\{1,\cdots,N\}$, and $\mathcal{I}_{B}=\{1,\cdots,M\}$ denote, respectively, the sets of AP indices and BS indices, respectively.
AP deployment and BS deployment are, respectively, defined by $P=(p_1,\cdots,p_N)$ and $Q=(q_1,\cdots,q_M)$, where $p_n\in\Omega$ is the location of AP $n$ and the location of $q_m\in\Omega$ is BS $m$. An AP partition $\R^{A}=\{R^{A}_n\}_{n\in \mathcal{I}_{A}}$ is a collection of disjoint subsets of $\Re^2$ whose union is $\Omega$.
Let $T: \mathcal{I}_{A}\to \mathcal{I}_{B}$ be an index map for which $T(n)=m$ if and only if AP $n$ is connected to BS $m$. A continuous and differentiable function $f(\cdot): \Omega^2\to\Re^+$ is used to denote the density of the data rate from the densely distributed sensors \cite{YTH}.

Usually, BSs have access to reliable energy sources and their energy consumption is not the main concern in this paper. Therefore, we focus on the energy consumption of sensors and APs. In fact, the energy consumed by sensors and APs comes from three parts: (i) Sensors transmit bit streams to APs; (ii) APs transmit bit streams to BSs; (iii) APs receive bit streams from sensors.
{\color{black}The transmitting powers (Watt) of nodes, e.g., sensors and APs, mainly depend on two factors: (i) The per-bit transmission energy (Joules/bit); (ii) The average bit rate (bits/s) going through the device.}
The average transmitting power of AP $n$ is defined as $\mathcal{P}_{t_n}^{A}=E_{t_n}^{A}\Gamma_n, n\in \mathcal{I}_{A}$, where $E_{t_n}^{A}$ is the per-bit transmission energy of AP $n$ and $\Gamma_n$ is the average bit rate transmitted from AP $n$ to the corresponding BS. $\Gamma_n$ is also the bit rate gathered from sensors in $R^{A}_n$. Therefore, we model the bit rate as $\Gamma_n=\int_{R^{A}_n}\!f(w)dw$.
Note that the signal-to-noise ratio (SNR) at the receiver depends on the transmitting distance and the antenna gain. In order to achieve the required SNR thresholds at the receivers, the per-bit transmission energy $E_{t_n}^{A}$ should be set to a value that is determined by the distance, the antenna gain, and the SNR threshold \cite{RH}. More realistically, the transmission power is proportional to distance squared or another power of distance in the presence of obstacles \cite{KA}.
Taking path-loss into consideration, it is reasonable to model the per-bit transmission energy from AP $n$ to BS $T(n)$ by
$E_{t_n}^{A} = \xi_{nT(n)}\|p_n-q_{T(n)}\|^{\alpha}$, where $\|\cdot\|$ denotes the Euclidean distance, $\xi_{nT(n)}$ is a constant determined by the antenna gain of AP $n$ and the SNR threshold of BS $T(n)$, and $\alpha\in[2,4]$ is the path-loss parameter.
We consider an environment without obstacles, i.e., $\alpha=2$. Moreover, we focus on homogeneous sensors, APs, and BSs. Hence, we consider identical sensor antenna gains, identical AP antenna gains, identical AP SNR thresholds, and identical BS SNR thresholds. Without loss of generality, we set $\xi_{nT(n)}=\xi$.

Let us now discuss the energy consumption at receivers. According to \cite{YTH}, the power at the receiver of AP $n$ can be modeled as $\mathcal{P}_{r_n}^{A}=\rho_n\sum\Gamma_n, n\in \mathcal{I}_{A}$, where $\rho_n$ is a constant determined by the SNR threshold of AP $n$. For homogeneous APs, $\rho_n=\rho,\forall n\in \mathcal{I}_{A}$. The power consumption at receivers is a constant and thus ignored in our objective function. Therefore,
to consider the total energy consumption, we use the sum of the AP transmission powers, $P^A_{t_n}$, which is calculated as
\begin{equation*}
\mathcal{P}^{A}(P,Q,\R^{A},T)=\sum_{n=1}^{N}\int_{R^{A}_n}\xi\|p_n-q_{T(n)}\|^{2}f(w)dw.
\end{equation*}
Similarly, the sum of the sensor transmission powers is calculated as
\begin{equation*}
\mathcal{P}^{S}(P,\R^{A})=\sum_{n=1}^{N}\int_{R^{A}_n}\eta\|p_n-w\|^{2}f(w)dw,
\end{equation*}
where $\eta$ is a constant determined by the antenna gain of sensors and the SNR threshold of the corresponding AP.

Furthermore, the sensor energy consumption is generally more crucial than the AP energy consumption because there are more sensors and it is more difficult to replenish energy in sensors. To compensate for any preference between energy consumptions in APs and sensors, we multiply the AP power by a non-negative weight $\gamma$. Let $\beta=\frac{\gamma\xi}{\eta}$.
We define the objective function (distortion) to be minimized as
\begin{equation}
D(P,Q,\R^{A},T)=\frac{1}{\eta}\left[\mathcal{P}^{S}(P,\R^{A})+\gamma \mathcal{P}^{A}(P,Q,\R^{A},T)\right]=\sum_{n=1}^{N}\int_{R^{A}_n}\left[\|p_n-w\|^{2}+\beta\|p_n-q_{T(n)}\|^{2}\right]f(w)dw.
\label{cost2}
\end{equation}
Our main goal is to minimize the distortion defined in (\ref{cost2}) by choosing the optimal AP deployment and the optimal BS deployment. Although this optimization problem is motivated by two-tiered WSNs, it can be applied to the two-tiered cellular networks, where sensor nodes can be other wireless devices, such as laptops and cell phones.
\section{The Best Possible Distortion for the Two-tiered WSNs with One BS}\label{sec:opt1}
The node deployment problem in the two-tiered WSNs can be interpreted as a two-tiered quantization problem whose reproduction points are APs and BSs.
Similarly, if one only considers the energy consumption of sensors, the corresponding optimization problem becomes a ``regular" quantization problem with distortion
\begin{equation}
 D_r(X,\R) = \sum_{n=1}^N\int_{R_n}\|x_n-w\|^2f(w)dw.
\end{equation}
Let $(X^*,\R^*)=\arg\min_{(X,\R)}D_r(X,\R)$ be the optimal regular quantizer.
In some cases \cite{YTH,HLJW,WP,JSYKT}, only one BS or fusion center is deployed to collect data from WSNs. We assume one BS and multiple APs in the following proposition.
\begin{prop}
For a two-tiered WSN with one BS, we have the following:
\begin{enumerate}[label=(\roman*)]
\item The optimal BS location is the centroid of the target region, i.e., $q^*=\frac{\int_\Omega wf(w)dw}{\int_\Omega f(w)dw}$.
\item The optimal AP locations of the two-tiered WSNs are linear transformations of the optimal reproduction points $X^*=(x^*_1,\dots,x^*_N)$, i.e., $p^*_n=\frac{x^*_n+\beta q^*}{1+\beta}, n\in\mathcal{I}_{A}$.
\item The optimal AP partition is the same as the optimal regular quantizer partition $\R^*=(R^*_1\dots,R^*_N)$, i.e., $\R^{A*}=\R^*$. \item The best possible distortion is $$\frac{1}{1+\beta}D_r(X^*,\R^*)+\frac{\beta}{1+\beta}\int_{\Omega}\|w-q^*\|^2f(w)dw.$$
\end{enumerate}
\label{Transformation}
\end{prop}

\begin{IEEEproof}
When there is only one BS, all APs transfer data to the unique BS located at $q$, and the index map is simply given by $T(n)=1, n\in \mathcal{I}_{A}$. The corresponding distortion is
\begin{align*}
D(P,Q,\R^{A},T) \!= \! \sum_{i=1}^N\int_{R^{A}_n}\!\!\left[\|p_n\!-\!w\|^2 \!+\! \beta\|p_n\!-\!q\|^2\right]f(w)dw.
\end{align*}
Since
\begin{align} \textstyle
\|p_n\!-\!w\|^2\!+\!\beta\|p_n\!-\!q\|^2\!=\!(1\!+\!\beta)\bigl\|p_n\!-\!\frac{(w+\beta q)}{1+\beta}\bigr\|^2 \!+\! \frac{\beta\|w-q\|^2}{1+\beta},
\label{eq}
\end{align}
we obtain
\begin{equation}
\!\!\! D(P,Q,\R^{A},T) = \frac{1}{1+\beta}\sum_{n=1}^N\int_{R^{A}_n}\!\!\!\|\left((1+\beta)p_n\!-\!\beta q\right)-w\|^2f(w)dw+ { \frac{\beta}{1+\beta}}\int_{\Omega}\|w-q\|^2f(w)dw,
\label{costHR1}
\end{equation}
The first term in (\ref{costHR1}) is the distortion of a regular quantizer with linear transformation of its reproduction points (AP locations). The minimum value of the first term is $D_r(X^*,\R^*)$ and can be achieved by choosing the optimal AP deployment $P$ for any BS location $q$. On the other hand, the second term in (\ref{costHR1}) is the distortion of another quantizer whose reproduction point is the BS, and is independent of the choice of AP locations and partition cells. In other words, the second term only depends on the BS location $q$. As a result, one can optimize (\ref{costHR1}) by finding the optimal $q^*$ to minimize the second term and then calculate the optimal AP deployment $P^*$ for $q^*$.
By parallel axis theorem, the second term achieves the minimum if and only if the BS is placed at the geometric centroid $c=\frac{\int_\Omega wf(w)dw}{\int_\Omega f(w)dw}$ of $\Omega$, which proves (i).
The best possible distortion is then the summation of $\frac{1}{1+\beta}D_r(X^*,\R^*)$ and $\frac{\beta}{1+\beta}\int_{\Omega}\|w-c\|f(w)dw$, which proves (iv).
The two-tiered quantizer achieves this minimum when $q^*=c$, $x^*_n=\left((1+\beta)p^*_n-\beta q^*\right), n\in \mathcal{I}_{A}$, and $\R^{A}=\R^*$, which proves (ii) and (iii).
\end{IEEEproof}

By Proposition \ref{Transformation}, one can obtain the optimal solution for the two-tiered quantization by shrinking the optimal reproduction points of the regular quantizer, {\color{black}which have been well studied in \cite{gray1,EH},} towards the corresponding BSs.
{\color{black}For example, consider one BS, $n$ APs, and a uniform distribution over the $1$-dimensional target region $\Omega = [s,t]$. The reproduction points and the cells of an optimal regular quantizer are given by $x^*_n=s+\frac{(2n-1)(t-s)}{2N}, n\in \mathcal{I}_A$, and $R^{*}_n=\left[s+\frac{(n-1)(t-s)}{N},s+\frac{n(t-s)}{N}\right], n\in \mathcal{I}_A$.
Therefore, in the two-tiered WSNs, the optimal BS location is
\begin{equation}
q^* = \frac{t+s}{2}.
\end{equation}
The optimal AP locations are
\begin{equation}
p_n^* = s+\frac{\beta(t-s)}{2(1+\beta)}+\frac{(2n-1)(t-s)}{2N(1+\beta)}, n\in \mathcal{I}_{A},
\end{equation}
and the optimal AP cell partitions are
\begin{equation}
R^{A*}_n=\left[s+\frac{(n-1)(t-s)}{N}, s+\frac{n(t-s)}{N}\right], n\in \mathcal{I}_{A}.
\end{equation}
The corresponding minimum distortion is
\begin{equation}
\frac{(t-s)^2}{12(1+\beta)N^2}+\frac{\beta(t-s)^2}{12(1+\beta)}.
\label{1dopt}
\end{equation}}
In particular, for $\Omega=[-\frac{1}{2},\frac{1}{2}]$ with $1$ BS and $4$ APs, the optimal BS location is $0$, the optimal cells are $R^{A*}_n=\left[\frac{n-3}{4}, \frac{n-2}{4}\right], n\in \mathcal{I}_{A}$, the optimal AP locations are $p_n^* = \frac{2n-5}{16}, n\in \mathcal{I}_{A}$, and the best possible distortion is $\frac{17}{384}$. The optimal deployment and partition are shown in Fig. \ref{deploymentexample}.

{\setlength{\floatsep}{10pt plus 1.0pt minus 20.0pt}
\begin{figure}[!htb]
\setlength\abovecaptionskip{0pt}
\setlength\belowcaptionskip{-20pt}
\centering
\includegraphics[width=7in]{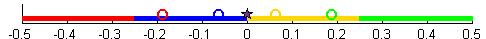}
\captionsetup{justification=justified}
\caption{An example of the best deployment and partition. The optimal AP locations are denoted by circles. The optimal partition cells are denoted by intervals. The optimal BS location is denoted by the star. Each AP and its corresponding cell are illustrated by the same color.}
\label{deploymentexample}
\end{figure}}

\section{The Optimal Deployment in two-tiered WSNs with Multiple BSs}\label{sec:opt2}
In this section, we extend the analysis of the optimal deployment to WSNs with multiple BSs. Distortion in this general case is determined by (i) AP deployment, (ii) BS deployment, (iii) AP cell partition, and (iv) Clustering (or the index map from APs to BSs). Before we discuss the optimal AP and BS deployment, we need to know (a) the best index map $T$ given $P$, $Q$, and $\R^{A}$, and (b) the best AP partition $\R^{A}$ given $P$, $Q$, and $T$.

The index map $T$ only influences the second term in (\ref{cost2}).
To minimize the second term, each AP should transfer data to the closest BS. Thus, the best index map is $\T_{[P,Q]}(n)=\arg\min_{m}\|p_n-q_m\|$.
However, given $P$, $Q$, and $T=\T_{[P,Q]}$, AP cell partition $\R^{A}$ affects both terms in (\ref{cost2}). The best AP cell partitions, called the energy Voronoi diagrams (EVDs), are
\begin{equation*}
V_n^E(P,Q) = \{w|\|p_n-w\|^{2}+\beta\|p_n-q_{\T_{[P,Q]}(n)}\|^{2}\le\|p_l-w\|^{2}+\beta\|p_l-q_{\T_{[P,Q]}(l)}\|^{2},\forall l\ne n\}, n\in \mathcal{I}_{A}.
\end{equation*}
Now, let $\V^E(P,Q)\!\!\!=\!\!\!\{V_n^{E}(P,Q)\}_{n\in \mathcal{I}_{A}}$ be the energy Voronoi partition.
Putting the best index map $\T_{[P,Q]}$ and the best AP partition $\V^E(P,Q)$ into (\ref{cost2}), the distortion is
\begin{equation}
\widetilde{D}(P,Q) = D(P,Q,\V^E(P,Q),\T_{[P,Q]})
\!=\!\sum_{n=1}^{N}\!\int_{V_n^{E}(P,Q)}\!\!\!\!\!\bigl(\|p_n\!-\!w\|^{2}\!+\!\beta\min_m\!\|p_n\!-\!q_{m}\|^{2}\bigr)f(\!w\!)dw
\label{OBJ}
\end{equation}
We have the following result, whose proof is provided in Appendix \ref{appendix1}.
\begin{prop}
Let $\alpha=2$ and $N>M$. The necessary conditions for the optimal deployment in a two-tiered WSN are
\begin{align}
p^*_n & =\frac{ c_n(P^*,Q^*)+\beta q^*_{\T_{[P^*,Q^*]}(n)}}{1+\beta}, n\in \mathcal{I}_{A},
\label{optP} \\
q_m^* & =\frac{\sum_{n:\T_{[P^*,Q^*]}(n)=m}c_n(P^*,Q^*)v_n(P^*,Q^*)}{\sum_{n:\T_{[P^*,Q^*]}(n)=m}v_n(P^*,Q^*)}, m\in \mathcal{I}_{B},
\label{optQ}
\end{align}
where $p_n^*$ is the optimal location for AP $n$ and $q_m^*$ is the optimal location for BS $m$, $\T_{[P,Q]}(n)=\arg\min_{m}\|p_n-q_m\|$ is the best index map, $v_n(P^*,Q^*)=\int_{V^E_n(P^*,Q^*)}f(w)dw$ is the volume of $V^E_n(P^*,Q^*)$ and $c_n(P^*,Q^*)$ is the geometric centroid of $V^E_n(P^*,Q^*)$.
\label{prop1}
\end{prop}

According to (\ref{optP}), the optimal location of AP $n$ should be on the segment $\overline{c_n(P,Q)q_m}$, where AP $n$ is connected with BS $m$. According to (\ref{optQ}), the best location of BS $m$ should be the geometric centroid of the $m^{th}$ cluster region $\bigcup_{n:T(n)=m}V^E_n(P^*,Q^*)$. Obviously, the optimal deployment and the optimal partition in Proposition \ref{Transformation} also satisfy the necessary conditions in Proposition \ref{prop1}. In the next section, using Proposition \ref{prop1}, we design Lloyd-like algorithms to determine the optimal deployment.

First, note that when the AP cell partition is fixed, the geometric centroid $c_n, n\in \mathcal{I}_{A}$, and the volume of the cells $v_n, n\in \mathcal{I}_{A}$, are fixed. Second, the index map $T^*$ represents the best connection between APs and BSs (or clustering) if and only if $P$ and $Q$ are given.
We now find the optimal
deployment, the optimal partition, the optimal index map and the best possible distortion for a uniform density in one-dimensional space.
\begin{theorem}
Let $\Omega=[s,t]$ with length $\mu(\Omega)=t-s$. Also, let
\begin{align}
\ell_a &= \left(\beta+\lceil\tfrac{N}{M}\rceil^{-2}\right)^{-\frac{1}{2}}, & \ell_b &= \left(\beta+\lfloor\tfrac{N}{M}\rfloor^{-2}\right)^{-\frac{1}{2}} \\
M_a & =(N\mbox{ mod } M), & M_b & =M-(N\mbox{ mod } M).
\end{align}
Then, given a uniform distribution on $\Omega$ with $M$ BSs and $N$ APs, the minimum distortion is
\begin{align}
\label{optD}
\frac{\mu^2(\Omega)}{12(1+\beta)}\left(M_a\ell_a+M_b\ell_b\right)^{-2}.
\end{align}
The minimum is achieved
if and only if
\begin{enumerate}[label=(\roman*)]
\item $M_a$ of the clusters are associated with $\lceil\frac{N}{M}\rceil$ APs each and have length $\ell_a \mu(\Omega) / (M_a\ell_a + M_b\ell_b)$,\label{I},
\item $M_b$ of the clusters are associated with $\lfloor\frac{N}{M}\rfloor$ APs each and have length $\ell_b \mu(\Omega) / (M_a\ell_a + M_b\ell_b)$,\label{II},
\item BSs are deployed at the centroids of the cluster regions,\label{III},
\item AP cells are uniform partitions of the cluster\label{IV},, and
\item AP $n$ is deployed at $\frac{c_n+\beta q_{T(n)}}{1+\beta}$, $n\in \mathcal{I}_{A}$, where $c_n$ is the geometric centroid of the AP $n$'s cell.\label{V},
\end{enumerate}
\label{2dcost}
\end{theorem}
The proof is provided in {\color{black}Appendix \ref{appendix2}.}\\
{\color{black}In particular, when $K=\frac{N}{M}$ is a positive integer,
the optimal BS locations are
\begin{equation}\small
q^*_m = s+\frac{(2m-1)(t-s)}{2M}, m\in \mathcal{I}_{B},
\label{2doptq}
\end{equation}
and the optimal index map is
\begin{equation}\small
T^*(n)=\left\lceil\frac{n}{K}\right\rceil, n\in \mathcal{I}_{A}.
\label{2doptT}
\end{equation}
The optimal AP locations are
\begin{equation}\small
p_n^* = s+\frac{(t-s)}{(1+\beta)}\left(\frac{(2n-1)}{2N}+\beta\frac{(2\lceil\frac{n}{K}\rceil-1)}{2M}\right), n\in \mathcal{I}_{A},
\label{2doptp}
\end{equation}
and the optimal AP cell partitions are
\begin{equation}\small
R^{A*}_i=\left[s+\frac{(n-1)(t-s)}{N}, s+\frac{n(t-s)}{N}\right], n\in \mathcal{I}_{A}.
\label{2doptR}
\end{equation}
The corresponding minimum distortion is
\begin{equation}\small
\frac{(t-s)^2}{12(1+\beta)M^2}\left(\frac{1}{K^2}+\beta\right).
\label{2dopt}
\end{equation}
}
\section{One-tiered and Two-tiered Lloyd Algorithms}\label{sec:algorithm}
We design two algorithms, one-tiered Lloyd (OTL) and two-tiered Lloyd (TTL) algorithms, to minimize the distortion in two-tiered WSNs. First, we quickly review the conventional Lloyd algorithm. Lloyd Algorithm has two basic steps in each iteration: (i) The node deployment is optimized while the partitioning is fixed; (ii) The partitioning is optimized while the node deployment is fixed.  As shown in \cite{GJ}, Lloyd algorithm, which provides good performance and is simple enough to be implemented distributively, can be used to solve regular quantizers or one-tiered node deployment problems.
However, the conventional Lloyd Algorithm cannot be applied to two-tiered WSNs where two kinds of nodes are deployed. Therefore, we introduce two Lloyd-based algorithms to solve the optimal deployment problem in two-tiered WSNs.
\subsection{One-tiered Lloyd Algorithm}
OTL Algorithm combines two independent Lloyd Algorithms. Using the Lloyd algorithm, an $M$-level regular quantizer is designed and its reproduction points are used as $Q$. Another $N$-level regular quantizer is designed and its partition is used as $\R^{A}$. The index map is determined by $T(n)=\arg\min_{m}\|p'_n-q_m\|$ and the deployment $P$ is determined by $p_n=\min_{m}\frac{p'_n+\beta q_m}{1+\beta}, n\in \mathcal{I}_{A}$, where $p'_n$ is the $n^{th}$ reproduction point obtained by the $N$-level quantizer.  Using Proposition \ref{Transformation}, it is easy to show that, for the networks with one BS, the distortion of OTL Algorithm converges to the minimum as long as the second Lloyd Algorithm provides the optimal $N$-level quantizer.

OTL Algorithm combines two independent Lloyd Algorithms. The two Lloyd Algorithms are, respectively, used to implement an $M$-level regular quantizer and an $N$-level regular quantizer. The $\R^{A}$ is determined as the partition of the second quantizer, and $Q$ is determined by the reproduction points of the first quantizer. The index map is determined by $T(n)=\arg\min_{m}\|p'_n-q_m\|$ and the deployment $P$ is determined by $p_n=\min_{m}\frac{p'_n+\beta q_m}{1+\beta}, n\in \mathcal{I}_{A}$, where $p'_n$ is the $n^{th}$ reproduction point obtained by the second quantizer.  Using Proposition \ref{Transformation}, it is easy to show that, for the networks with $1$ BS, the distortion of OTL Algorithm converges to the minimum as long as the second Lloyd Algorithm provides the optimal $N$-level quantizer.
\subsection{Two-tiered Lloyd Algorithm}
Before we introduce the details of TTL Algorithm, we introduce two concepts: (i) AP local distortion, (ii) BS local distortion. The AP local distortion is defined as
\begin{align*}
D_n^{A}(P,Q,\!\R^{A}\!,\!T)\!\!=\!\!\!\!\int_{R_n^{A}}\!\!\!\left[\|p_n\!-\!w\|^{2}\!+\!\beta\|p_n\!-\!q_{T(n)}\|^{2}\right]\!f(w)dw.
\end{align*}
The total distortion is the summation of these AP local distortions, i.e., $D(P,Q,\R^{A},T)\!=\!\sum_{n=1}^ND^{A}_n(P,Q,\R^{A},T)$.
Similarly, for $\mathcal{N}_m \triangleq \{n: T(n) = m\}$,  the BS local distortion is defined as
\begin{align*}\small
D_m^{B}(P,Q,\!\R^{A},\!T\!)\!\!=\!\!\!\!\sum_{n\in\mathcal{N}_m}\!\!\!\int_{\!R_n^{A}}\!\!\!\!\!\left(\|p_n\!-\!w\|^{2}\!+\!\beta\|p_n\!-\!q_{m}\|^{2}\right)\!\!f(\!w\!)dw.
\end{align*}
The total distortion is the sum of the BS local distortions, i.e., $D(P,Q,\R^{A},T)=\sum_{m=1}^MD^{B}_m(P,Q,\R^{A},T)$. Let $c_n$ and $v_n$ be, respectively, the geometric centroid and the volume of the current AP cell partition.
Now, we provide the details of TTL Algorithm.
The TTL algorithm iterates over four steps: (i) AP $n$ moves to $\frac{c_n+\beta q_{T(n)}}{1+\beta}$; (ii) AP partitioning is done by assigning the corresponding EVD to each AP node; (iii) BS $m$ moves to $\frac{\sum_{n\in\mathcal{N}_m}p_nv_n}{\sum_{n\in\mathcal{N}_m}v_n}$; (iv) Clustering is done by assigning the nearest BS to each AP.

In what follows, we show that the distortion of TTL Algorithm converges.
First, due to the parallel axis theorem and (\ref{eq}), the local distortion of AP $n$ can be rewritten as
\begin{equation}
D_n^{A}(P,Q,\!\R^{A},\!T\!)=\frac{1}{1+\beta}\int_{R^{A}_n}\|c_n-w\|^{2}f(w)dw+(1+\beta)\|p_n-\widehat p_n\|^2v_n+\frac{\beta}{1\!+\!\beta}\int_{R^{A}_n}\!\!\!\|w\!-\!q_{T(n)}\|^{2}f(w)dw,
\label{APlocalcost}
\end{equation}
where $\widehat p_n=\frac{c_n+\beta q_{T(n)}}{1+\beta}$. When $Q$, $\R^{A}$ and $T$ are given, the first term and the third term of (\ref{APlocalcost}) are constants. In other words, the AP local distortion becomes a function of $\|p_n-\widehat p_n\|$. Therefore, Step (i) does not increase the AP local distortions and then the total distortion.
Second, given $P$, $Q$ and $T$, EVDs minimize the total distortion, indicating that the total cost is not increased by Step (ii).
Third, observe that for $\widehat q_m=\frac{\sum_{n\in\mathcal{N}_m}p_nv_n}{\sum_{n\in\mathcal{N}_m}v_n}$, we have $\sum_{n\in\mathcal{N}_m}v_n\|p_n-q_m\|^2 = \sum_{n\in\mathcal{N}_m}v_n\left[\|p_n-\widehat q_m\|^2+\|q_m-\widehat q_m\|^2\right]$.
Therefore, the local distortion of BS $m$ can be rewritten as
\begin{equation}
D^{B}_n(P,Q,\!\R^{A},\!T)=\sum_{n\in\mathcal{N}_m}\int_{R_n^{A}}\|p_n-w\|^{2}f(w)dw+ \beta\left(\sum_{n\in\mathcal{N}_m}v_n\right)\|q_m-\widehat q_m\|^2+\beta\sum_{n\in\mathcal{N}_m}\left(v_n\|p_n-\widehat q_m\|^2\right).
\label{BSlocalcost}
\end{equation}
When $P$, $\R^{A}$ and $T$ are given, the first term and the third term in (\ref{BSlocalcost}) are constants. In other words, the BS local distortion becomes a function of $\|q_m-\widehat q_m\|$. Therefore, Step (iii) does not increase the BS local distortions and then the total distortion.
Last, given $P$, $Q$ and $\R^{A}$, $T^*(n)=\arg\min_{m}\|p_n-q_m\|$ minimizes the total distortion, indicating that the total distortion is not increased by Step (iv).
In other words, the algorithm generates  is a positive non-increasing sequence of distortion values and therefore will converge.

\section{Performance Evaluation}\label{sec:simulation}
We provide the simulation results in two two-tiered WSNs: (i) WSN1: A two-tiered WSN including 1 BS and 20 APs; (ii) WSN2: A two-tiered WSN including 4 BSs and 20 APs. The target region is set to $\Omega=[0,10]^2$ and $\beta$ is set to $1$. The traffic density function is the sum of five Gaussian functions of the form $5exp(0.5(-(x-x_{center})^2-(y-y_{center})^2))$, where centers $(x_{center},y_{center})$ are (8,1), (4,9), (7.6,7.6), (9.4,5) and (2,2). Totally, we generate 50 initial AP and BS deployments on $\Omega$ randomly, i.e., every node location is generated with uniform distribution on $\Omega$. For each initial AP and BS deployments, we connect every AP to its closest BS and then assign the corresponding EVD to the AP node. The maximum number of iterations is set to 100. BSs and APs are denoted, respectively, by colored five-pointed stars and colored circles. The corresponding geometric centroid of AP cells are denoted by colored asterisks. Each BS and its connected APs form a cluster. To make clusters more visible, the symbols in the same cluster are filled with the same color.

Figs. \ref{InitialDeploymentInWSN1}, \ref{OTLDeploymentInWSN1}, and \ref{TTLDeploymentInWSN1} show one example of the initial and the final deployments of the two algorithms, OTL Algorithm and TTL Algorithm, in WSN1. For the initial deployment in Fig. \ref{InitialDeploymentInWSN1}, 12 APs have non-empty cells and make contributions to the data collection. After running the OTL and TTL algorithms, all APs have non-empty cells.
Compared to the random node placement with the corresponding optimal AP partitioning and clustering, the OTL and TTL algorithms save, respectively, 57.00\% and 56.78\% of the weighted power.
{
\begin{figure}[!htb]
\centering
\subfloat[]{\includegraphics[width=2.1in]{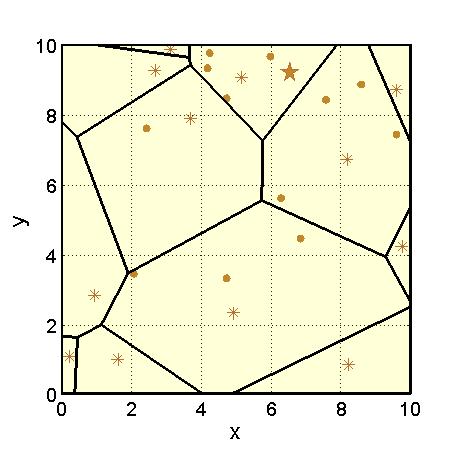}
\label{InitialDeploymentInWSN1}}
\hfil
\subfloat[]{\includegraphics[width=2.1in]{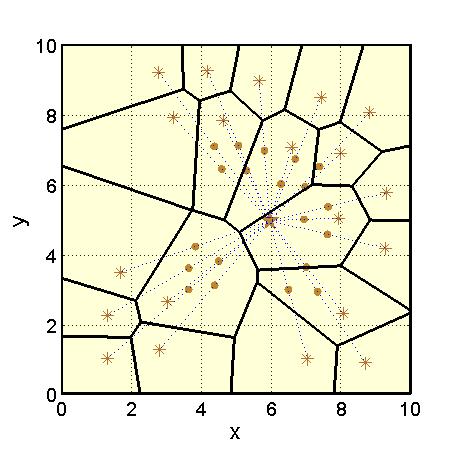}
\label{OTLDeploymentInWSN1}}
\hfil
\subfloat[]{\includegraphics[width=2.1in]{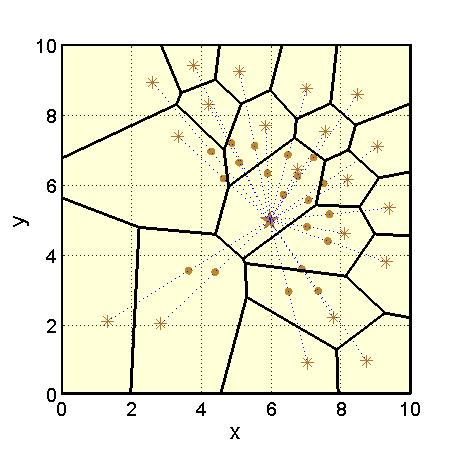}
\label{TTLDeploymentInWSN1}}
\captionsetup{justification=justified}
\caption{AP and BS deployments in WSN1. (a) The initial deployment and partitions. (b) The final deployment and  partitions of OTL Algorithm. (c) The final deployment and partitions of TTL Algorithm.}
\label{PerformanceInWSN1}
\end{figure}}

Figs. \ref{InitialDeploymentInWSN2}, \ref{OTLDeploymentInWSN2}, and \ref{TTLDeploymentInWSN2} show similar results for WSN2. For the initial deployment in Fig. \ref{InitialDeploymentInWSN2}, 16 APs have non-empty cells and make contributions to the data collection. Compared to the random node placement with the corresponding optimal AP partitioning and clustering, the OTL and TTL algorithms save, respectively, 79.46\% and 80.19\% of the weighted power.
\begin{figure}[!htb]
\centering
\subfloat[]{\includegraphics[width=2.1in]{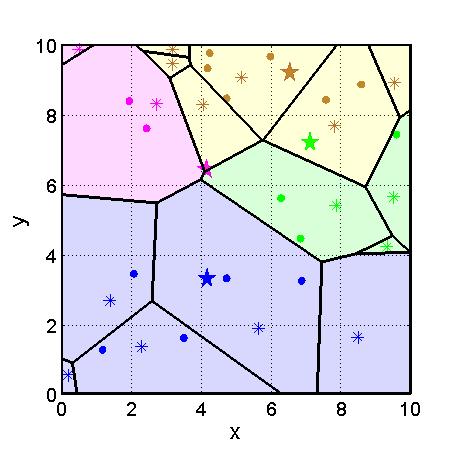}
\label{InitialDeploymentInWSN2}}
\hfil
\subfloat[]{\includegraphics[width=2.1in]{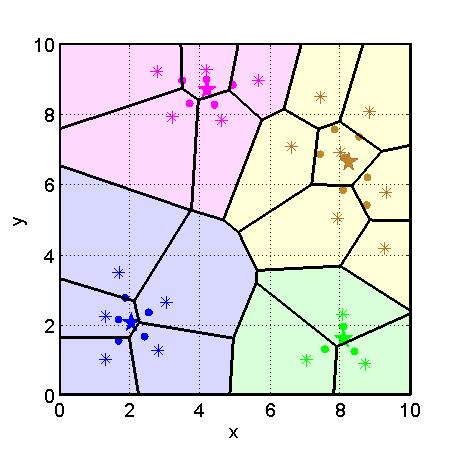}
\label{OTLDeploymentInWSN2}}
\hfil
\subfloat[]{\includegraphics[width=2.1in]{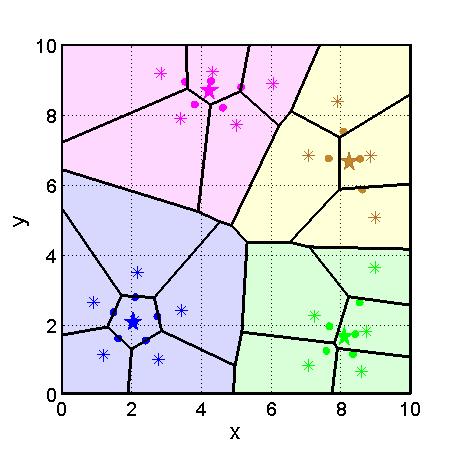}
\label{TTLDeploymentInWSN2}}
\captionsetup{justification=justified}
\caption{AP and BS deployments in WSN2. (a) The initial deployment and partitions. (b) The final deployment and partitions of OTL Algorithm. (c) The final deployment and partitions of TTL Algorithm.}
\label{DeploymentInWSN2}
\end{figure}

As discussed in Section \ref{sec:model}, the distortion is the weighted power. For each randomly generated initial deployment, we calculate the initial weighted power and the weighted powers after running the two algorithms starting with these initial deployments. With these simulation results, we calculate the percentage of saved power for each initial deployment and then the averaged percentage of saved power over 50 initial deployments for each algorithm. Our statistical results show that, on average, OTL algorithm saves, respectively, $53.61\%$ and $79.29\%$ of the weighted power in WSN1 and WSN2. Similarly, on average, TTL algorithm saves, respectively, $53.71\%$ and $79.16\%$ of the weighted power in WSN1 and WSN2.

\section{Conclusions}\label{sec:conclusion}
A two-tiered wireless sensor network which collects data from a large-scale wireless sensor network to base stations through access points is discussed in this paper. We studied the energy consumption on such two-tiered WSNs and provided the corresponding objective function. Different from one-tiered WSN, both the AP deployment and the BS deployment are taken into consideration.
The necessary condition for optimal deployment implies that every AP location
should be deployed between the centroid of its cell
and its associated BS. By defining an appropriate distortion measure, we
proposed one-tiered Lloyd (OTL) and two-tiered Lloyd (TTL) algorithms to minimize the distortion.
Our simulation results show that OTL and TTL algorithms greatly save the weighted power or energy in a two-tiered WSNs.

\appendices
\section{Proof of Proposition 2}\label{appendix1}
Let $(P^*,Q^*,\R^{A*},T^*)$ be an optimal solution for (\ref{cost2}).
As we show at the beginning of Sec. \ref{sec:opt2}, given the optimal deployment $(P^*,Q^*)$, the optimal partition and the optimal index map are, respectively, $\R^{A*}=V^E(P^*,Q^*)$ and $T^*=\T_{[P^*,Q^*]}$.
Thus, the optimal geometric centroid and the optimal (Lebesgue) measure of $\R^{A*}$ can be represented as $c_n(P^*,Q^*)$ and $v_n(P^*,Q^*)$, where $c_n(P,Q)=\frac{\int_{V^E(P,Q)}wf(w)dw}{\int_{V^E(P,Q)}f(w)dw}$ and $v_n(P,Q)=\int_{V^E(P,Q)}f(w)dw$.
According to the parallel axis theorem, given the optimal partition $V^E(P^*,Q^*)$ and the optimal index map $\T_{[P^*,Q^*]}$, the objective function in (\ref{cost2}) can be expressed as
\begin{equation}
\begin{aligned}
D(\!P,\!Q,\!V\!^E\!\!(\!P^*\!,\!Q^*\!),\!\T_{[P^*\!\!,Q^*\!]})
&{=}\!\sum_{n=1}^{N}\!\left[\!\int_{V^E_n\!(P^*\!\!,Q^*\!)}\!\!\!\!\!\!\!\!\!\!\!\!\!\!\|c_n(\!P^*\!\!,Q^*\!)\!-\!w\|^{2}f(w)dw\!+\!\|p_n\!-\!c_n(\!P^*\!\!,Q^*\!)\!\|^2\!v_n(\!P^*\!\!,Q^*\!)\!+\!\beta\|p_n\!-\!q_{\T_{[\!P^*\!\!,Q^*\!]}}\!\|^{2}\!v_n(\!P^*\!\!,Q^*\!)\!\right].
\label{cost3}
\end{aligned}
\end{equation}
The partial derivatives of (\ref{cost3}) are
\begin{equation}
\frac{\partial D(P,Q,V^E(P^*,Q^*),\T_{[P^*,Q^*]})}{\partial p_n} = 2\left[(p_n-c_n(P^*,Q^*))+\beta(p_n-q_{\T_{[P^*,Q^*]}(n)})\right]v_n(P^*,Q^*), n\in\mathcal{I}_{A},
\end{equation}
and
\begin{equation}
\frac{\partial D(P,Q,V^E(P^*,Q^*),\T_{[P^*,Q^*]})}{\partial q_m} = \sum_{n:\T_{[P^*,Q^*]}(n)=m}2\beta(q_{m}-p_n)v_n(P^*,Q^*), m\in\mathcal{I}_{B}, .
\end{equation}
Since (\ref{cost3}) is a convex function of $P$ and $Q$, the optimal deployment $(P^*,Q^*)$ is unique and satisfies\\
$\frac{\partial D(P,Q,V^E(P^*,Q^*),\T_{[P^*,Q^*]})}{\partial p_n}|_{(P^*, Q^*)}=0$ and $\frac{\partial D(P,Q,V^E(P^*,Q^*),\T_{[P^*,Q^*]})}{\partial q_m}|_{(P^*, Q^*)}=0$.
Solving for $p_n^{*}$ and $q_m^{*}$, we obtain
\begin{equation}
p^*_n = \frac{c_n(P^*,Q^*)^*+\beta q^*_{\T_{[P^*,Q^*]}(n)}}{1+\beta}, n\in\mathcal{I}_{A}
\label{res1}
\end{equation}
and
\begin{equation}
q^*_m = \frac{\sum_{n:\T_{[P^*,Q^*]}(n)=m}p^*_nv_n(P^*,Q^*)}{\sum_{n:\T_{[P^*,Q^*]}(n)=m}v_n(P^*,Q^*)}, m\in\mathcal{I}_{B}
\label{res2}
\end{equation}
Substituting (\ref{res1}) to (\ref{res2}), we have
\begin{equation}
q_m^* =\frac{\sum_{n:\T_{[P^*,Q^*]}(n)=m}c_n(P^*,Q^*)v_n(P^*,Q^*)}{\sum_{n:\T_{[P^*,Q^*]}(n)=m}v_n(P^*,Q^*)}, m\in\mathcal{I}_{B}.
\label{res3}
\end{equation}

\section{Proof of Theorem 1}\label{appendix2}
{\color{black}
Before we discuss the best possible distortion in the uniformly distributed 1-dimensional space, we need to present the following concepts and lemmas.
Let $\mu(W)$ be the (Lebesgue) measure of the set $W$.
Let $d(N,\Omega)=\min_{x_1,\ldots,x_N}\int_{\Omega}\min_{n}\|x_n-w\|^2\frac{dw}{\mu(\Omega)}$ be the minimum distortion of the regular $N$-level  quantizer for a uniform distribution on $\Omega \subset\mathbb{R}$.
\begin{lemma}
We have $d(N,\Omega)\ge \frac{\mu(\Omega)^2}{12N^2}$ with equality if and only if $\Omega$ is the union of $N$ disjoint intervals, each with measure $\frac{\mu(\Omega)}{N}$.
\label{dna}
\end{lemma}

\begin{proof}
Let $p_1,\dots,p_N\in\Omega$, and $R_1,\ldots,R_N\subset\mathbb{R}$ respectively denote the reproduction points and the quantization cells of the optimal regular quantizer that achieves $d(N, \Omega)$. Note that $p_n$ is the centroid of $R_n$. We have
\begin{align}
\label{lemma1eqs1} d(N,\Omega) & = \sum_{n=1}^N \int_{R_n} \|p_n-w\|^2\frac{dw}{\mu(\Omega)} \\ \label{lemma1eqs2}  & = \sum_{n=1}^{N}\frac{\mu(R_n)}{\mu(\Omega)}d(1,R_n)\\
 \label{lemma1eqs3}  & \ge \sum_{n=1}^{N}\frac{1}{\mu(\Omega)}\frac{\mu(R_n)^3}{12}\\
\label{lemma1eqs4}  &\ge \frac{1}{12\mu(\Omega)}\left(\sum_{n=1}^{N}\mu(R_n)\right)^3N^{-2}\\
\label{lemma1eqs5}  &=\frac{\mu(\Omega)^2}{12N^2},
\end{align}
where (\ref{lemma1eqs2}) follows since for any $p_n\in\mathbb{R}$, we have $d(1,R_n) =  \int_{R_n} \|p_n-w\|^2\frac{dw}{\mu(R_n)}$ by definition, and (\ref{lemma1eqs3}) follows since for any $A\subset\mathbb{R}$, we have
$d(1,A)\ge \frac{\mu(A)^2}{12}$ with equality if and only if $A$ is an interval \cite{CS}. Also, (\ref{lemma1eqs4})  is the reverse H\"{o}lder's inequality, and (\ref{lemma1eqs5})  follows since $\sum_{n=1}^{N}\mu(R_n)=\mu(\Omega)$. Note that (\ref{lemma1eqs3})  is an equality if and only if $\Omega_n$s are intervals, and (\ref{lemma1eqs4})  is an equality if and only if $\mu(\Omega_n)=\frac{\mu(\Omega)}{N}, \forall n$. Therefore, (\ref{lemma1eqs5}) can be achieved if and only if $\Omega$ is the union of disjoint intervals with the same measure $\frac{\mu(\Omega)}{N}$.
\end{proof}

\begin{lemma}
Let $N$ and $M$ be two positive integers such that $N\ge M$. We define a function $D^{LB}(e_1,\dots,e_M)=\left(\sum_{m=1}^{M}\left(\beta+\frac{1}{e_m^2}\right)^{-\frac{1}{2}}\right)^{-2}$ with the domain $R_c=\{(e_1,\dots,e_M)|\sum_{m=1}^Me_m=N,e_m\in\mathbb{N},\forall m\}$, where $\beta$ is a non-negative constant. Let $M_a =(N\mbox{ mod } M)$ and $M_b = M - M_a$.
Then, $D^{LB}(e_1,\dots,e_M)$ attains the unique minimum
\begin{equation}\small
\left(M_a\left(\beta+\frac{1}{\lceil\frac{N}{M}\rceil^2}\right)^{-\frac{1}{2}}+M_b\left(\beta+\frac{1}{\lfloor\frac{N}{M}\rfloor^2}\right)^{-\frac{1}{2}}\right)^{-2}
\label{min}
\end{equation}
where $M_a$ of the $e_m$s are equal to $\lceil\frac{N}{M}\rceil$ and $M_b$ of the $e_m$s are equal to $\lfloor\frac{N}{M}\rfloor$. In particular, when $M_a=0$, $D^{LB}(e_1,\dots,e_M)$ attains the unique minimum $\left(\frac{\beta}{M^2}+\frac{1}{N^2}\right)$ at $\left(\frac{N}{M},\dots,\frac{N}{M}\right)$.
\label{DLB}
\end{lemma}

\begin{proof}
Let $\mathbf{e}=(e_1,\dots,e_M)\in R_c$, $D^{UB}(\mathbf{e})=\sum_{m=1}^{M}\left(\beta+\frac{1}{e_m^2}\right)^{-\frac{1}{2}}$. Minimizing $D^{LB}$ is equivalent to maximizing $D^{UB}$. Let $i,j\in\{1,\ldots,M\}$ be two arbitrary indices. Without loss of generality, suppose $e_i\ge e_j$. Let $\delta=e_i-e_j$ and $\zeta=e_i+e_j$.
We have
\begin{equation}\small
\widetilde D^{UB}_{ij}\!(\delta)\!=\!D^{UB}\!\!\left(e_1,\dots,\frac{\zeta+\delta}{2},\dots,\frac{\zeta-\delta}{2},\dots,e_m\right)\!\!=\!\!\sum\limits_{k\ne i,j}\left(\beta\!+\!\frac{1}{e_k^2}\right)^{\!\!\!-\frac{1}{2}}\!\!\!\!\!+\left(\beta\!+\!\frac{1}{\left(\frac{\zeta+\delta}{2}\right)^2}\right)^{\!\!\!-\frac{1}{2}}\!\!\!\!\!+ \left(\beta\!+\!\frac{1}{\left(\frac{\zeta-\delta}{2}\right)^2}\right)^{\!\!\!-\frac{1}{2}},
\label{delta}
\end{equation}
and therefore,
\begin{equation}\small
\frac{\partial \widetilde D^{UB}_{ij}(\delta)}{\partial\delta}=\frac{1}{2}\left[\left(1+\beta\left(\frac{\zeta+\delta}{2}\right)^2\right)^{-\frac{3}{2}}-\left(1+\beta\left(\frac{\zeta-\delta}{2}\right)^2\right)^{-\frac{3}{2}}\right].
\end{equation}
Let $g(y)=\frac{1}{2}y^{-\frac{3}{2}}, y\in(0,\infty)$, $y(x)=1+x^2$, $x\in[0,\infty)$. Since $e_i$ and $e_j$ are non-negative and $e_i\ge e_j$, we have $\delta\ge0$, $\zeta\ge0$, and then $x_1=\frac{\zeta+\delta}{2}\ge\frac{\zeta-\delta}{2}=x_2$. Consequently, $y(x_1)\ge y(x_2)>0$, and thus, $\frac{\partial\widetilde D^{UB}(\delta)}{\partial\delta}=g(y(x_1))-g(y(x_2))\le0$ with equality if and only if $\delta=0$. Therefore, $\widetilde D^{UB}_{ij}(\delta)$ is a decreasing function for non-negative continuous $\delta$.

Now, let $\mathbf{e}^*=(e^*_1,\dots,e^*_M)\triangleq\arg\min_{(e_1,\dots,e_M){\in R_c}}D^{LB}(e_1,\dots,e_M)$ be a minimizer of $D^{LB}$ on $R_c$, and $\hat \delta\triangleq\min_{i\ne j}|e^*_i-e^*_j|$ be the minimum difference among $e^*_m$s. Since $e^*_i$s are positive integers, we have $\hat \delta\in\mathbb{N}$.
In what follows, we show that $\hat \delta\in\{0,1\}$. Suppose $\hat \delta\ge 2$. Then, we can find two indices $i,j\in\{1,\ldots,M\}$ such that $\delta=e^*_i-e^*_j\ge 2$. Let $\mathbf{e}'=(e^*_1,\dots,e'_i,\dots,e'_j,\dots,e^*_M)$ be a new solution where $e'_i=e^*_i-1$, and $e'_j=e^*_j+1$. We have $\delta'=e'_i-e'_j=e^*_i-e^*_j-2=\delta-2<\delta$.  Since $\widetilde D^{UB}_{ij}(\delta)$ is a monotonically decreasing function for non-negative continuous $\delta$, we have $\widetilde D^{UB}_{ij}(\delta')>\widetilde D^{UB}_{ij}(\delta)$ where $\delta'$ and $\delta$ are non-negative integers. Thus,  we have $D^{UB}(\mathbf{e}')>D^{UB}(\mathbf{e}^*)$ which contradicts the optimality of $\mathbf{e}^{*}$.

Therefore, $\hat \delta\in\{0,1\}$, and $e^*_m$s can thus assume at most 2 distinct values. Suppose $M_1$ of the $e^*_m$s are equal to $h$ and $M_2$ of the $e^*_m$s are equal to $h+1$, where $h\geq 0$ is an integer.
It is self-evident that at least one of $M_1$ or $M_2$ should be positive.
Without loss of generality, suppose $M_1>0$ and $M_2\ge0$.
Since $M_1+M_2=M$ and $M_1>0$, we have $0<M_1\le M$ and $0\le M_2<M$.
From the equalities $M_1+M_2=M$ and $M_1h+M_2(h+1)=N$, we obtain $Mh+M_2=N$. Solving the system $Mh+M_2=N$ and $0\le M_2<M$, we have $h=\lfloor\frac{N}{M}\rfloor$ and $M_2=N\mod M=M_a$. Finally, using the equality $M_1+M_2=M$, we can determine $M_1=M-(N\mod M)=M_b$.
\end{proof}
}
Now, we have enough tools to derive the best possible distortion in the uniformly distributed 1-dimensional space.
{\color{black}
Let $\mathcal{N}_m=\{n|T(n)=m\}$ be the set of APs connected to the $m^{th}$ BS, and $N_m$ be the number of elements in $\mathcal{N}_m$. Let $W_m=\bigcup_{n\in \mathcal{N}_m}R^{A}_n$ be the $m^{th}$ cluster region.
We have
\begin{align}
D(P,Q,\R^{A},T)
=&\sum_{m=1}^M\sum\limits_{n\in \mathcal{N}_m}\int_{R^{A}_n}\left(\|p_n-w\|^2+\beta\|p_n-q_m\|^2\right)dw\\
=&\sum_{m=1}^M\sum\limits_{n\in \mathcal{N}_m}\int_{R^{A}_n}\left(\frac{1}{1+\beta}\|(1+\beta)p_n-\beta q-w\|^2+\frac{\beta}{1+\beta}\|w-q\|^2\right)dw\\
\ge&\sum_{m=1}^M\left[\frac{1}{1+\beta}d(N_m,W_m)+\frac{\beta}{1+\beta}d(1,W_m)\right]\frac{\mu(W_m)}{\mu(\Omega)}\label{LB0}\\
\ge&\frac{1}{12(1+\beta)\mu(\Omega)}\sum_{m=1}^M\mu^3(W_m)\left(\beta+\frac{1}{N_m^2}\right)\label{LB1}\\
\ge&\frac{1}{12(1+\beta)\mu(\Omega)}\left(\sum_{m=1}^{M}\mu(W_m)\right)^{3}\left(\sum_{m=1}^{M}\left(\beta+\frac{1}{N_m^2}\right)^{-\frac{1}{2}}\right)^{-2}\label{LB2}\\
=&\frac{\mu^2(\Omega)}{12(1+\beta)}\left(\sum_{m=1}^{M}\left(\beta+\frac{1}{N_m^2}\right)^{-\frac{1}{2}}\right)^{-2}\label{LB3}\\
\ge&\frac{\mu^2(\Omega)}{12(1+\beta)}\min_{\begin{subarray}{c}N_1,\dots,N_M\in\mathcal{N}\\\sum_{m=1}^MN_m=N\end{subarray}}\left(\sum_{m=1}^{M}\left(\beta+\frac{1}{N_m^2}\right)^{-\frac{1}{2}}\right)^{-2}\label{LB4}
\end{align}
where the first equality follows from (\ref{eq}), the first inequality follows from the definition of $d(N,\Omega)$, the second inequality follows from Lemma \ref{dna}, and the third inequality is the reverse H\"{o}lder's inequality. All these inequalities can be made tight with a specific choice of $p_n$, $q_m$, $W_m$s and $N_m$s. In fact, by Proposition \ref{prop1}, (\ref{LB0}) is an equality if and only if $p_n=\frac{c_n+\beta q_{T(n)}}{1+\beta}$ and $q_m$ is the centroid of $W_m$, indicating \ref{III} and \ref{V} in Theorem \ref{2dcost}.
Also, according to Lemma \ref{dna}, (\ref{LB1}) is an equality if and only if $W_m, m\in \mathcal{I}_{\mathcal B},$ are intervals, and $W_m$ is uniformly divided into $N_m$ intervals. Therefore, \ref{IV} in Theorem \ref{2dcost} is proved. According to the reverse H\"{o}lder's inequality, (\ref{LB2}) is an equality if and only if $\exists\tau>0, \mu(W_m)=\tau\left(\beta+\frac{1}{N_m^2}\right)^{-\frac{1}{2}}, \forall m\in \mathcal{I}_{\mathcal B}$. Moreover, the sum of these measures is $\mu(\Omega)$, i.e., $\sum_{j=1}^{M}\mu(W_j)=\mu(\Omega)$. Therefore, the corresponding measure of the $m^{th}$ cluster region is
\begin{equation}\small
\mu(W_m)=\frac{N_m\left(1+\beta N_m^2\right)^{-\frac{1}{2}}\mu(\Omega)}{\sum_{j=1}^{M}N_j\left(1+\beta N_j^2\right)^{-\frac{1}{2}}}.
\label{Wm}
\end{equation}
Note that (\ref{LB3}) is a function of $(N_1,\dots,N_M)$ and (\ref{LB4}) is just the minimum of (\ref{LB3}). Therefore, the last inequality is an equality when we properly select the variables $N_1,\dots,N_M$. Obviously, the above equality conditions are compatible, i.e., all equality conditions can be satisfied simultaneously. Therefore, (\ref{LB4}) is an achievable lower bound, indicating the minimum distortion. The last thing is to determine the optimal $(N_1,\dots,N_M)$ that attains the minimum of (\ref{LB3}).
By Lemma \ref{DLB}, (\ref{LB3}) attains (\ref{LB4}), if and only if $M_a$ of the $N_m$s are equal to $\lceil\frac{N}{M}\rceil$ and $M_b$ of the $N_m$s are equal to $\lfloor\frac{N}{M}\rfloor$.
Substituting the optimal values for $N_ms$ to (\ref{LB4}), we obtain the minimum distortion formula in (\ref{optD}).
Substituting the optimal values for $N_m$s to (\ref{Wm}), we get \ref{I} and \ref{II} in Theorem \ref{2dcost}.
}
\section*{Acknowledgment}

This work was supported in part by the DARPA GRAPHS program Award N66001-14-1-4061.



%

\end{document}